\documentclass[11pt]{llncs}

\usepackage[english]{babel}
\usepackage[T1]{fontenc}
\usepackage[latin1]{inputenc}
\usepackage{times}
\usepackage{amssymb,amsmath,amscd,latexsym,graphicx}
\usepackage{gastex}
\usepackage[usenames]{color}
\usepackage{lineno}

\newcommand{\set}[1]{\{ #1 \}}
\newcommand{\PSPACE}{\mathrm{PSPACE}}
\newcommand{\PTIME}{\mathrm{PTIME}}
\newcommand{\cocompl}{\mathrm{co}}

\newcommand{\LTL}{\mathrm{LTL}}

\newcommand{\FPT}{\mathrm{FPT}}
\newcommand{\APTIME}{\mathrm{APTIME}}
\newcommand{\NP}{\mathrm{NP}}
\newcommand{\QBF}{\mathrm{QBF}}
\newcommand{\SAT}{\mathrm{SAT}}
\newcommand{\ksurdeux}{\lfloor k/2 \rfloor}
\newcommand{\Reach}{\mathrm{Reach}}
\newcommand{\GenReach}{\mathrm{GenReach}}
\newcommand{\Attr}{\mathrm{Attr}}

\newcommand{\next}{\nu}
\newcommand{\M}{\mathcal{M}}
\newcommand{\up}{\mu}
\newcommand{\N}{\mathbb{N}}

\newcommand{\VE}{V_{\circ}}
\newcommand{\VA}{V_{\Box}}
\newcommand{\W}{\mathcal{W}}
\newcommand{\Aa}{\mathcal{A}}
\newcommand{\WE}{\W_{E}}
\newcommand{\WA}{\W_{A}}
\newcommand{\G}{\mathcal{G}}
\newcommand{\binomksurdeuxk}{\binom{k}{\ksurdeux}}

\def\qed{\rule{0.4em}{1.4ex}}

\begin{document}
\title{The surprising complexity of\\
generalized reachability games}
\author{Nathana\"el Fijalkow \inst{1,2} \and Florian Horn \inst{1}}
\institute{LIAFA \\ CNRS \& Universit\'e Denis Diderot - Paris 7, France \\
\email{\{nath,florian.horn\}@liafa.jussieu.fr}\\
\and \'ENS Cachan \\ \'Ecole Normale Sup\'erieure de Cachan, France}


\maketitle

\begin{abstract}
Games on graphs provide a natural and powerful model for reactive systems.
In this paper, we consider generalized reachability objectives,
defined as conjunctions of reachability objectives.
We first prove that deciding the winner in such games is $\PSPACE$-complete, 
although it is fixed-parameter tractable with the number of reachability objectives as parameter.
Moreover, we consider the memory requirements for both players
and give matching upper and lower bounds on the size of winning strategies.
In order to allow more efficient algorithms,
we consider subclasses of generalized reachability games.
We show that bounding the size of the reachability sets gives
two natural subclasses where deciding the winner can be done 
efficiently.
\end{abstract}

\section{Introduction}

\medskip\noindent{\bf Graphs games.}
Our purpose is to study reactive systems by abstracting them into graphs games:
a state of the system is represented by a vertex in a finite directed graph,
and a transition corresponds to an edge.
If in a given state, the controller can choose the evolution of the system,
then the corresponding vertex is controlled by the first player, Eve.
Otherwise, the system evolves in an uncertain way: 
we consider the worst-case scenario where
a second player, Adam, controls those states.
To a run of the system corresponds a play on the game:
we put a pebble in the initial vertex, 
then Eve and Adam move this pebble along the edges,
constructing an infinite sequence.
The specification of the system gives an objective Eve tries to ensure on this sequence.
In order to synthesize a controller,
we are interested in two questions: whether Eve wins
in the game, and what resources are needed to construct a winning strategy (see~\cite{GTW02-Dagstuhl} for more details).

\medskip\noindent{\bf System specifications.}
To specify properties of a system, we construct a set of infinite sequences
representing the correct behaviors of the system.
From an infinite sequence we extract finite information
to decide whether the run it represents meet the specification.
For instance, considering the set of vertices visited infinitely often
allows to specify the classical $\omega$-regular properties,
\textit{e.g} B\"uchi, parity, Streett, Rabin and M\"uller objectives.
Other informations can be carried out, as for instance the set of vertices visited with positive frequency~\cite{TBG09},
or the order in which the vertices are visited for specifying $\LTL$ objectives~\cite{KPV07-CAV,HTW08-ATVA,Zim}.
In this work, we observe the set of vertices visited at least once,
which allows to specify reachability objectives, 
also called weak objectives~\cite{NeumannSW02,StaigerW74,Mostowski91,KupfermanVW00} 

\medskip\noindent{\bf Generalized reachability objectives.}
The (simple) reachability objective requires, given a subset of vertices $F$, that a vertex from $F$ is reached.
Reachability objectives only specifies that one property
(represented by $F$) is satisfied along the run.
We allow more properties to be specified by using generalized reachability objectives,
defined as conjunctions of $k$ reachability objectives.
In this context, a reachability objective is often referred as a color: 
a generalized reachability objective is then to see each of the $k$ colors at least once.

\section{Definitions}

The games we consider are played on an \emph{arena} $\Aa = (V,(\VE,\VA),E)$, which consists of a finite graph $(V,E)$ 
and a partition $(\VE,\VA)$ of the vertex set $V$: 
a vertex is controlled by Eve if it belongs to $\VE$
and by Adam if it belongs to $\VA$.
Vertices from $\VE$ are depicted by a circle, and vertices from $\VA$ by a square.
We denote by $n$ the number of vertices and $m$ the number of edges.
Playing consists in moving a pebble along the edges: the pebble is placed on the initial vertex $v_0$,
then the player who controls the vertex chooses an edge and sends the pebble along this edge
to the next vertex.
From this infinite interaction results a \emph{play} $\pi$, which is an infinite sequence of vertices
$v_0, v_1, \ldots$ where for all $i$, we have $(v_i,v_{i+1}) \in E$, \textit{i.e} $\pi$ is an infinite path in the graph.
We denote by $\Pi$ the set of all plays, and define \emph{objectives}
for a player by giving a set of winning plays $\Phi \subseteq \Pi$. 
The games are zero-sum, which means that if Eve has the objective $\Phi$, then Adam has the objective 
$\Pi \setminus \Phi$ (the objectives are opposite).
Formally, a \emph{game} is given by a couple $\G = (\Aa,\Phi)$ where $\Aa$ is an arena 
and $\Phi$ an objective.

A \emph{strategy} for a player is a function that prescribes, given a finite history of the play, the next move.
Formally, a \emph{strategy} for Eve is a function $\sigma : V^* \cdot \VE \to V$ 
such that for a finite history $w \in V^*$ and a current position $v \in \VE$, the prescribed move is legal,
\textit{i.e} along an edge: $(v,\sigma(w \cdot v)) \in E$.
Strategies for Adam are defined similarly, and usually denoted by $\tau$.
Once a game $\G = (\Aa,\Phi)$, a starting vertex $v_0$ and strategies $\sigma$ for Eve and $\tau$ for Adam are fixed, 
there is a unique play denoted by $\pi(v_0,\sigma,\tau)$, which is said to be winning for Eve if it belongs to $\Phi$.
The sentence ``Eve wins from $v_0$'' means that she has a winning strategy from $v_0$, that is a strategy $\sigma$ 
such that for all strategy $\tau$ for Adam, the play $\pi(v_0,\sigma,\tau)$ is winning.
The first natural problem we consider is to ``solve the game'', that is given a game $\G$ and a starting vertex $v_0$, 
to decide whether Eve wins from $v_0$.
We denote by $\WE(\G)$ the winning positions of Eve, 
that is the set of vertices from where Eve wins (also referred as winning set), and analogously $\WA(\G)$ for Adam.
We can prove that in generalized reachability games,
we have $\WE(\Aa,\Phi) \cup \WA(\Aa,\Phi) = V$: from any vertex, either of the two players has a winning strategy.
We say that the games are \textit{determined}.

The strategies as defined in their full generality above are infinite objects.
Indeed, in this general setting, to pick the next-move, Eve considers the whole history of the play,
whose size grows arbitrarily.
A nicer setting, giving rise to finitely-representable objects, is to define strategies relying on memory structures.
Formally, a \emph{memory structure} $\M = (M, m_0,\up)$ for an arena $\Aa$ consists of a set $M$ of memory states, 
an initial memory state $m_0 \in M$, and an update function $\up: M \times E \to M$. 
A memory structure is similar in fashion to an automaton synchronized with the arena: 
it starts from $m_0$ and reads the sequence of edges produced by the arena.
Whenever an edge is taken, the current state is updated using the update function $\up$.
A strategy relying on a memory structure $\M$, whenever it picks the next move, 
considers only the current vertex and the current memory state: 
it is thus given by a next-move function $\next: \VE \times M \to V$.
Formally, given a memory structure $\M$ and a next-move function $\next$, 
we can define a strategy $\sigma$ for Eve by $\sigma(w \cdot v) = \next(v, \up^*(w \cdot v))$.
(The update function can be extended to a function $\up^*: V^+ \to M$ by defining $\up^*(v) = m_0$ and 
$\up^* (w \cdot u \cdot v) = \up(\up^*(w \cdot u), (u,v))$.)
A strategy with memory structure $\M$ has finite memory if $M$ is a finite set.
It is \emph{memoryless}, or \emph{positional} if $M$ is a singleton: in this case, the choice for the next move
only depends on the current vertex.
Note that a memoryless strategy can be described as a function $\sigma: \VE \to V$.

We can make the synchronized product explicit: an arena $\Aa$ and a memory structure $\M$ for $\Aa$ induce 
the expanded arena $\Aa \times \M = (V \times M, (\VE \times M, \VA \times M), E \times \up)$ where $E \times \up$ is defined by:
$((v,m), (v',m')) \in E'$ if $(v,v') \in E$ and $\up(m,(v,v')) = m'$.
There is a natural one-to-one mapping between plays in $\Aa$ and in $\Aa \times \M$,
and also from memoryless strategies in $\Aa \times \M$
to strategies in $\Aa$ using $\M$ as memory structure.
It follows that if a player has a memoryless winning strategy for the arena $\Aa \times \M$, then
he has a winning strategy using $\M$ as memory structure for the arena $\Aa$.
This \textit{key} property will be used later on.

A \emph{reachability objective} requires that a vertex from a given subset of vertices $F$ is reached:
$\Reach(F) = \set{v_0, v_1, v_2 \ldots \mid \exists p \in \N, v_p \in F} \subseteq \Pi$.
Games in the form $\G = (\Aa,\Reach(F))$ are called reachability games.
To determine whether Eve wins a reachability game, we compute the reachability set attractor.
We define the sequence $(\Attr_i(F))_{i \geq 0}$:
$$\begin{array}{llll}
\Attr_0(F) & = & F \\
\Attr_{i+1}(F) & = & \Attr_i(F) & \cup \quad \set{u \in \VE \mid \exists (u,v) \in E, v \in \Attr_i(F)} \\
 & & & \cup \quad \set{u \in \VA \mid \forall (u,v) \in E, v \in \Attr_i(F)} \\
\end{array}$$
Then $\Attr(F)$ is the limit of the non-decreasing sequence $(\Attr_i(F))_{i \geq 0}$.
We can prove that $\WE(\Aa,\Reach(F))$ is exactly $\Attr(F)$.

\medskip\noindent{\bf Generalized reachability objectives.}
A \emph{generalized reachability objective} requires that 
each of the given $k$ subsets of vertices $F_1,\ldots,F_k$ is reached:
$$\GenReach(F_1,\ldots,F_k) = \set{\pi \mid \forall i, \exists p_i \in \N, v_{p_i} \in F_i}.$$
Associating to each reachability objective a color, 
we can reformulate the generalized reachability objective:
it requires to see each of the $k$ colors at least once, in any order.
Games in the form $\G = (\Aa,\GenReach(F_1,\ldots,F_k))$ are called generalized reachability games. 
The special cases where in $\Aa$, $\VA$ (respectively $\VE$) is empty are called one-player 
(respectively opponent-player) generalized reachability games.

\begin{example}
We consider the arena drawn in Figure~\ref{fig:ex_gen_reach}. 
A generalized reachability game is defined by the objective
$\GenReach(\set{1,2},\set{3})$.
The central vertex is the initial one.
Eve tries to visit one of the two thick vertices and the dashed vertex.

\begin{figure}
\label{fig:ex_gen_reach}
\begin{center}
\begin{picture}(40,40)(0,0)
	\gasset{Nw=8,Nh=8}

  	\node[Nmarks=i,iangle=-90](1)(20,15){}
  	\rpnode[linewidth=0.5,polyangle=45](2)(35,0)(4,5){$1$}
  	\node[linewidth=0.5](3)(5,0){$2$}
  	\rpnode[dash={1.5}0,polyangle=45](4)(20,30)(4,5){$3$}

  	\drawedge(1,2){}
  	\drawedge(1,3){}
  	\drawedge(1,4){}
  	\drawedge[curvedepth=-2](2,3){}
  	\drawedge[curvedepth=-2](3,2){}
  	\drawedge[curvedepth=-2](2,4){}
  	\drawedge[curvedepth=2](3,4){}
	\drawloop(4){}
\end{picture}
\end{center}
\caption{An example of a generalized reachability game}
\end{figure}
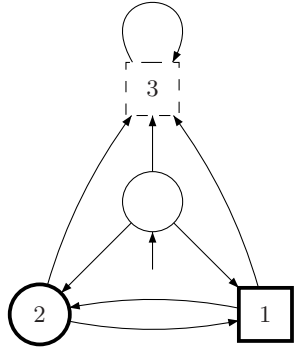
\end{example}

\medskip\noindent{\bf Contributions.}
Our contributions are as follows:
\begin{itemize}
	\item We first prove that deciding the winner in generalized reachability games is $\PSPACE$-complete.
Using the same ideas, we also show that the one-player restriction, where
all vertices belong to Eve, is $\NP$-complete,
and that the opponent-player restriction, where
all vertices belong to Adam, can be solved in polynomial time.
On the positive side, it is fixed-parameter tractable $(\FPT)$ with the number $k$ of colors as parameter.

	\item We study the size of the winning strategies for both players:
we prove matching upper and lower bounds, \textit{i.e}
in any arena, if Eve has a winning strategy, then she has a winning strategy that uses $2^k - 1$ memory states,
and there is an arena where Eve wins but there are no winning strategies with less than $2^k - 1$ memory states,
and similarly for Adam with the bound $\binomksurdeuxk$.

	\item We then consider the subclasses where we restrict the number of vertices sharing the same color
(in other words, the size of reachability sets).
This reveals a trichotomy: if three vertices are allowed to share the same color,
then deciding the winner is, as in the general case, $\PSPACE$-complete.
However, if each color appears only once, then the problem is polynomial.
If each color appears only twice, then the problem is polynomial for one-player games,
where Eve controls all vertices.
\end{itemize}

\medskip\noindent{\bf Outline.}
In section 3, we first study the complexity of solving generalized reachability games,
for two-player and one-player games,
and then give matching upper and lower bounds for the memory required.
In section 4, we consider the subclasses of games where the size of reachability set is restricted,
in order to find tractable subclasses.

\section{The complexity of generalized reachability games}

In this section we prove that the winner problem in generalized reachability games 
is $\PSPACE$-complete.
Our $\PSPACE$-hardness result follows from a reduction from $\QBF$ 
(evaluation of a quantified boolean formula in conjunctive normal form).
However, we show that solving generalized reachability games with few colors is easy,
as it is fixed-parameter tractable using the number of colors as parameter.

We then study one-player restrictions.
We prove that the one-player generalized reachability games are $\NP$-complete.
The other one-player restriction, opponent-player generalized reachability games, 
can be solved in polynomial time.

The last subsection investigates memory requirements for both players.
We present matching upper and lower bounds:
Eve needs $2^k - 1$ memory states and Adam $\binomksurdeuxk$,
where $k$ is the number of colors.

\subsection{$\PSPACE$-completeness of solving generalized reachability games}

As a first step we define a reduction from $\QBF$ to the winner problem of generalized reachability games.
Consider a quantified boolean formula
$$Q_1 x_1 \ Q_2 x_2 \ \ldots Q_n x_n \ \phi\ ,$$
where $\phi$ is a propositional formula in conjunctive normal form, \textit{i.e}
$$\phi = \bigwedge_{i \leq k} \ell_{i,1} \ \vee \ \ell_{i,2} \ \vee \ \ldots \ \vee \ \ell_{i,j_i}$$
and $\ell_{i,j}$ is either $x_i$ or $\neg x_i$ for some $i \leq n$.
We construct a generalized reachability game where Eve wins if and only if the formula is true. 
Intuitively, the two players will sequentially choose to assign values to variables, 
following the quantification order and starting from the outermost variable. 
Eve chooses existential variables and Adam chooses universal variables.
Formally, the game is as follows:
\begin{itemize}
	\item for each variable $x_i$, there are two vertices, $x_i$ and $\overline{x_i}$;
	\item for each variable $x_i$, there is a choice vertex $v_i$ which leads to $x_i$ and $\overline{x_i}$.
The choice vertex belongs to Eve if $x_i$ is existentially quantified, and to Adam if $x_i$ is universally quantified;
	\item for each variable $x_i$ with $i < n$, there are two edges from $x_i$ and $\overline{x_i}$ to the next choice vertex $v_{i+1}$;
	\item there is a sink $s$, and two edges from $x_n$ and $\overline{x_n}$ to $s$;
	\item for each clause $\set{\ell_{i,1},\ldots,\ell_{i,j_i}}$,
there is a reachability objective $F_i$ which contains the corresponding vertices;
	\item the generalized reachability objective is given by $\GenReach(F_1,\ldots,F_k)$.
\end{itemize}
The initial vertex is $v_1$.
There is a natural bijection between assignments of the variables and plays in this game;
and an assignment satisfies the formula $\phi$ if and only if 
the play satisfies the generalized reachability objective.
The evaluation order of the variables being the same in the formula and in the game,
we conclude that Eve has a winning strategy if and only if the formula is true.

\begin{example}
We consider the following quantified boolean formula
$$\forall x\ \exists y\ \forall z\ (x \vee \neg y) \wedge (\neg y \vee z)\ .$$
Figure~\ref{fig:example_red} shows the game built by the reduction.
The generalized reachability objective is
$\Reach(\set{x, \overline{y}}) \wedge \Reach(\set{\overline{y},z})$.
Thick vertices represent the first reachability objective
and dashed vertices the second one.

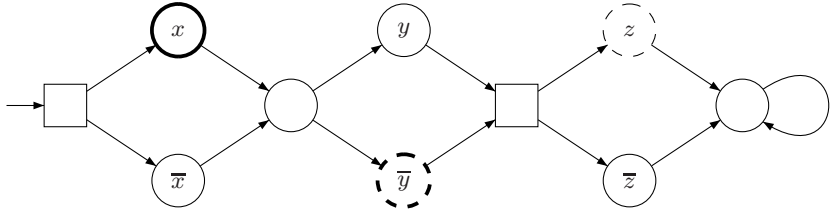
\begin{figure}
\label{fig:example_red}
\begin{center}
\begin{picture}(90,20)(5,0)
	\gasset{Nw=7,Nh=7}
 
	\rpnode[Nmarks=i,polyangle=45](1)(0,10)(4,4){}
	\node[linewidth=0.5](x)(15,20){$x$}
	\node(nx)(15,0){$\overline{x}$}

	\drawedge(1,x){}
	\drawedge(1,nx){}

	\node(2)(30,10){}
	\node(y)(45,20){$y$}
	\node[dash={1.5}0,linewidth=0.5](ny)(45,0){$\overline{y}$}

	\drawedge(x,2){}
	\drawedge(nx,2){}
	\drawedge(2,y){}
	\drawedge(2,ny){}

	\rpnode[polyangle=45](3)(60,10)(4,4){}
	\node[dash={1.5}0](z)(75,20){$z$}
	\node(nz)(75,0){$\overline{z}$}

	\drawedge(y,3){}
	\drawedge(ny,3){}
	\drawedge(3,z){}
	\drawedge(3,nz){}

	\node(e)(90,10){}
	\drawedge(z,e){}
	\drawedge(nz,e){}
	\drawloop[loopangle=0](e){}
\end{picture}
\end{center}
\caption{An example of the reduction from $\QBF$ to generalized reachability games.}
\end{figure}
\end{example}

\begin{theorem}[Complexity of generalized reachability games]
Solving generalized reachability games is $\PSPACE$-complete.
\end{theorem}

\begin{proof}
The previous reduction implies the $\PSPACE$-hardness.

Let us first make a simple observation: if Eve has a winning strategy, then she has a winning strategy 
that visits each reachability set within $n \cdot k$ steps.
Indeed, if she can enforce to visit a subset of vertices, then she can enforce it within $n$ steps.

Relying on this remark, we can simulate the game for up to $n \cdot k$ steps using an alternating Turing machine:
whenever a vertex belongs to Eve, the corresponding state is disjunctive, and it is conjunctive if the vertex belongs to Adam. 
A path of length $n \cdot k$ is accepted if it is winning, 
\textit{i.e} if it contains one vertex from each reachability set $F_i$. 
This machine accepts if and only if Eve wins, and works in polynomial time.
Since $\APTIME = \PSPACE$, the result follows.\hfill\qed
\end{proof}

\subsection{Parameterized complexity}

Solving generalized reachability games with few colors is easy:

\begin{theorem}[Generalized reachability games with $k$ colors]
\label{thm_two}
Solving generalized reachability games is fixed-parameter tractable ($\FPT$)
with the number of colors as parameter.
\end{theorem}

Roughly speaking, the only information needed during a play is 
the subset of reachability sets already visited.
We build a memory structure that keeps track of this information.
By constructing the product with this memory structure,
we turn a generalized reachability game into a (classical) reachability game.

\begin{proof}
We consider $\G = (G,\GenReach(F_1,\ldots,F_k))$ a generalized reachability game,
and $v_0$ a starting vertex.
The memory structure $\M$ is defined by $(2^{\set{1,\ldots,k}}, m_0, \up)$, 
where $m_0$ is $\set{i \mid v_0 \in F_i}$,
and $\up(S,(v,v')) = S \cup \set{i \mid v' \in F_i}$.
Let $F = \set{(\_,S) \mid S = \set{1,\ldots,k}}$:
a play for the generalized reachability game $\G$ from $v_0$ is winning 
if and only if 
it is winning for the reachability game $\G \times \M = (G \times \M,\Reach(F))$ 
from $(v_0,m_0)$.
Since deciding the winner in a reachability game can be done in linear time
using an attractor computation,
solving a generalized reachability game can be done in time $2^k \times O(n + m)$.\hfill\qed
\end{proof}

\subsection{Solving one-player restrictions}

\begin{theorem}[One-player restrictions]
Solving one-player generalized reachability games is $\NP$-complete.
Solving opponent-player generalized reachability games is polynomial.
\end{theorem}

\begin{proof}
We first deal with one-player generalized reachability games, where Eve controls all vertices.
In our previous reduction, consider the case where all variables in the original formula are quantified existentially. 
Then the problem corresponds to $\SAT$ (satisfiability of a boolean formula in conjunctive normal form), which is $\NP$-complete.
Resulting games are one-player games, \textit{i.e} all vertices belong to Eve,
hence solving one-player generalized reachability games is $\NP$-hard.

We describe a non-deterministic algorithm to solve these games in polynomial time. 
As noted before, if Eve wins, then she has a winning strategy that wins within $n \cdot k$ steps. 
The algorithm guesses a path of length $n \cdot k$ and checks whether it is winning.
It follows that solving one-player generalized reachability games is $\NP$-complete.

We now consider opponent-player generalized reachability games,
given by the objective $\GenReach(F_1,\ldots,F_k)$.
The winning set for Adam is $V \setminus \bigcap_i \Attr(F_i)$,
which can be computed in quadratic time.\hfill\qed
\end{proof}

\subsection{Memory requirements}

We first present upper bounds:

\begin{lemma}[Memory upper bounds]
\label{lem_mem_exp_upp}
For all generalized reachability games $\G = (G,\GenReach(F_1,\ldots,F_k))$,
\begin{itemize}
	\item if Eve wins, then she wins using a strategy with memory $2^k - 1$;
	\item if Adam wins, then he wins using a strategy with memory $\binomksurdeuxk$.
\end{itemize}
\end{lemma}

As in the proof for $\FPT$ membership, we make use of the memory structure 
$\M = (2^{\set{1,\ldots,k}}, m_0, \up)$, 
where $m_0$ is $\set{i \mid v_0 \in F_i}$,
and $\up(S,(v,v')) = S \cup \set{i \mid v' \in F_i}$.
Setting $F$ as $\set{(\_,S) \mid S = \set{1,\ldots,k}}$,
a play for the generalized reachability game $\G$ from $v_0$ is winning 
if and only if 
it is winning for the reachability game $\G \times \M = (G \times \M,\Reach(F))$ 
from $(v_0,m_0)$.

\begin{proof}
We consider $\G = (G,\GenReach(F_1,\ldots,F_k))$ a generalized reachability game,
and $v_0$ a starting vertex.
Since in the reachability game $\G \times \M$, each player has memoryless winning strategies, 
each player has in $\G$ a winning strategy using $\M$ as memory structure.

The memory set of $\M$ has size $2^k$.
In order to get the correct bounds for each player, we rely on two observations.
\begin{itemize}
	\item Eve does not need a specific memory state to remember that all colors
have been reached, as in this case, she has already won.
Thus, she can always win with $2^k-1$ memory states.
	\item If Adam wins in $\G \times \M$ from $(v,S)$ and $S' \subseteq S$,
then he wins from $(v,S')$ \textit{using the same strategy}.
Consider $v \in V$ a vertex in $\G$, 
and the set of subsets $S$ such that $(v,S)$ belongs to Adam's winning set.
Its maximal (with respect to inclusion) elements are incomparable, so there are at most $\binomksurdeuxk$, 
we denote them by $S_1(v),\ldots,S_p(v)$.
The idea is that from $v$, there are only $p$ different options Adam has to consider, namely 
$S_1(v),\ldots,S_p(v)$. 
Indeed, for any $S$ such that $(v,S)$ is winning for Adam,
there exists an $i$ such that $S \subseteq S_i(v)$, so Adam can forget $S$ and assume the current position is $(v,S_i(v))$.

We define a memory structure on the memory set $\set{1,\ldots,\binomksurdeuxk}$.
We aim at constructing a strategy that will ensure that after a finite play $\pi \cdot v$,
the memory state is an $i$ such that $S_i(v)$ 
contains the set of visited colors.
If the initial vertex is $v_0$, the initial memory state is an $i_0$ such that 
$S_{i_0}(v_0)$ contains $m_0$.
We define the update function: $\up(i,(v,v'))$ is a $j$ such that 
$S_j(v')$ contains $\up(S_i(v),v') = S_i(v) \cup \set{i \mid v' \in F_i}$.

Let us turn to the next-move function. 
Consider $(v,S)$ in Adam's winning set,
then there exists a transition to some $(v',S')$ also in Adam's winning set.
Applying this to $(v,S_i(v))$ such that $S \subseteq S_i(v)$, we get a vertex $v'$, and define $\next(v,i)$ to $v'$.
Playing this strategy, the above invariant is satisfied,
and thus ensures to stay forever in Adam's winning set, so it is winning.
The memory set contains $\binomksurdeuxk$ memory states.
\end{itemize}
\hfill\qed
\end{proof}

\begin{lemma}[Memory lower bounds for both players]
For all $k$, 
\begin{itemize}
	\item there exists $\G = (G,\GenReach(F_1,\ldots,F_k))$ a generalized reachability game, 
	where Eve needs $2^k - 1$ memory states to win;
	\item there exists $\G = (G,\GenReach(F_1,\ldots,F_k))$ a generalized reachability game, 
	where Adam needs $\binomksurdeuxk$ memory states to win.
\end{itemize}
\end{lemma}

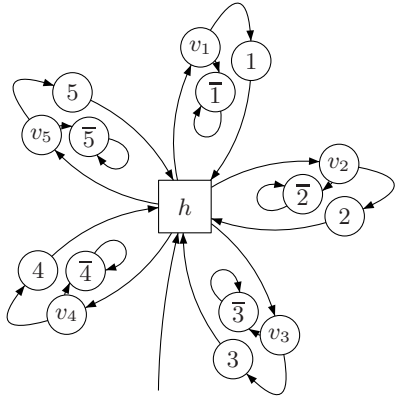
\begin{figure}
\begin{center}
\begin{picture}(10,45)(0,-23)
	\unitlength = 3.5mm
	\gasset{Nw=2,Nh=2,Nmr=2}

	\node[Nmr=0](coeur)(0,0){$h$}
	\node[Nframe=n](racine)(-1,-8){}
	\drawedge[curvedepth=.3](racine,coeur){}

	\gasset{Nw=1.5,Nh=1.5,Nmr=1.2}
	
	\node(a1)(0.59,6.05){$v_1$}
	\node(b1)(2.52,5.54){$1$}
	\node(c1)(1.16,4.35){$\overline{1}$}

	\drawedge[curvedepth=.7](coeur,a1){}
	\drawedge[curvedepth=2](a1,b1){}
	\drawedge[curvedepth=.7](b1,coeur){}
	\drawedge[curvedepth=.4](a1,c1){}
	\drawloop[loopangle=255,loopdiam=1](c1){}

	\node(a2)(5.86,1.62){$v_2$}
	\node(b2)(6.07,-0.37){$2$}
	\node(c2)(4.48,0.47){$\overline{2}$}

	\drawedge[curvedepth=.7](coeur,a2){}
	\drawedge[curvedepth=2](a2,b2){}
	\drawedge[curvedepth=.7](b2,coeur){}
	\drawedge[curvedepth=.4](a2,c2){}
	\drawloop[loopangle=186,loopdiam=1](c2){}

	\node(a3)(3.61,-4.89){$v_3$}
	\node(b3)(1.83,-5.80){$3$}
	\node(c3)(2.04,-4.01){$\overline{3}$}

	\drawedge[curvedepth=.7](coeur,a3){}
	\drawedge[curvedepth=2](a3,b3){}
	\drawedge[curvedepth=.7](b3,coeur){}
	\drawedge[curvedepth=.4](a3,c3){}
	\drawloop[loopangle=117,loopdiam=1](c3){}

	\node(ai)(-4.49,-4.11){$v_4$}
	\node(bi)(-5.58,-2.43){$4$}
	\node(ci)(-3.77,-2.45){$\overline{4}$}

	\drawedge[curvedepth=.7](coeur,ai){}
	\drawedge[curvedepth=2](ai,bi){}
	\drawedge[curvedepth=.7](bi,coeur){}
	\drawedge[curvedepth=.4](ai,ci){}
	\drawloop[loopangle=33,loopdiam=1](ci){}

	\node(ak)(-5.44,2.72){$v_5$}
	\node(bk)(-4.27,4.34){$5$}
	\node(ck)(-3.64,2.65){$\overline{5}$}

	\drawedge[curvedepth=.7](coeur,ak){}
	\drawedge[curvedepth=2](ak,bk){}
	\drawedge[curvedepth=.7](bk,coeur){}
	\drawedge[curvedepth=.4](ak,ck){}
	\drawloop[loopangle=324,loopdiam=1](ck){}
\end{picture}
\end{center}
\caption{A generalized reachability game where Eve needs $2^k - 1$ memory states to win}
\label{fig:flower}
\end{figure}

\begin{proof}
We first describe a generalized reachability game 
where Eve needs $2^k - 1$ memory states to win.
This example was proposed in~\cite{CHH} in a similar framework.
The arena is shown in Figure~\ref{fig:flower}, for $k = 5$.
A vertex labelled by $i$ belongs to $F_i$,
and a vertex labelled by $\overline{i}$ has all colors but $i$.
A play starts from the heart $h$;
first Adam chooses a petal $i$, 
then Eve chooses either to reach color $i$ before going back to the heart 
(the play goes on),
or to reach every colors but $i$ and to stop the play.
Eve wins with the following strategy: the first time Adam chooses the petal $i$, 
she goes back to the heart; the second time, she stops the play. 
This strategy uses $2^k$ memory states.
She can save one memory state by dropping the memory state corresponding to the case
where she saw each petal, as it is winning for her.
However, we show that there is no winning strategy for Eve 
with less than $2^k - 1$ memory states.
Let $\sigma$ a strategy using the memory structure $\M$
with less than $2^k - 1$ memory states, and $\next$ its next-move function.
For each memory state $m$, we consider
$S_m = \set{i \mid \next(v_i,\up(m,(h,v_i))) = \overline{i}}$,
the set of petals where Eve would stop the play if Adam chose them.
As there are less than $2^k - 1$ memory states, there is a strict subset $X$
of $\set{1,\ldots,k}$ which is not the stopping set of any memory state.
Adam can win against $\sigma$ by choosing, at each step, a petal in the symmetric
difference of $X$ and $S_m$, where $m$ is Eve's current memory under $\sigma$.
(Indeed, if Adam plays forever in $X$, then Eve will never stop the play and only colors from $X$ will be reached,
otherwise, whenever Eve stops the play, the last memory state from the heart was an $m$ such that $X \subset S_m$,
and the petal chosen is an $i$ that belongs to $S_m \setminus X$, hence that has never been reached.)

We now describe a generalized reachability game won by Adam, 
where he needs $\binomksurdeuxk$ memory states to win. 
Let $k = 2p + 1$.
A play consists in three steps:
first Eve chooses $p$ colors, 
then Adam chooses $p$ colors,
and third Eve chooses $p$ colors.
In order to win, Adam must visit exactly the same colors Eve visited (which requires $\binomksurdeuxk$ memory states),
otherwise at least $p+1$ colors have been visited when Eve plays for the second time,
and she can choose and visit the remaining colors that have not yet been visited.
\hfill\qed
\end{proof}

\section{Restrictions on the size of reachability sets}

The above section shows two different directions which make 
generalized reachability games hard: 
the first is the complexity of solving generalized reachability games ($\PSPACE$-complete),
and the second is the memory required to construct winning strategies for both players
(exponential in the number of colors).

In this section, we restrict the size of the reachability sets in order to find 
tractable subclasses of generalized reachability games.

Notice that our reduction from $\QBF$ only implies $\PSPACE$-hardness
when reachability sets have size at least three.
Indeed, note that in the reduction,
the size of a reachability set in the generalized reachability game
corresponds to the size of the corresponding clause of the formula.
Since the problem of evaluating a quantified boolean formula 
is polynomial if the formula has two variables per clause,
our reduction does not imply the $\PSPACE$-hardness of
solving generalized reachability games with reachability sets of size one or two.
This remark motivates our study of the subclasses of generalized reachability games where each color appears once,
and then where each color appears twice.

\subsection{Reachability sets of size one}

The case where reachability sets are singletons is polynomial:

\begin{theorem}[Generalized reachability games where reachability sets have size $1$]
\label{thm_singleton}
Solving generalized reachability games where reachability sets are singletons is in $\PTIME$.
\end{theorem}

\begin{proof}
We denote by $v_i$ the only vertex in $F_i$, for all $i$.
In this case, the generalized reachability objective can be expressed by $\bigwedge_{i \leq k} \Reach(v_i)$.
We will see that Eve wins if and only if 
the preorder defined by $v \preceq v'$ if $v \in \Attr(v')$ is total.
Intuitively, it means that a winning strategy prescribes: ``reach $v_{f(1)}$, then $v_{f(2)}$, and so on'', 
where $f$ is a permutation over $\set{1,\ldots,k}$.

We consider two cases:
\begin{itemize}
	\item If the preorder $\preceq$ is total over $\set{v_i \mid 1 \leq i \leq k}$, 
	then we show that $\WE$, set of winning positions for Eve, is $\cap_i \Attr(v_i)$. 
Let $v \in \cap_i \Attr(v_i)$ and $f$ a permutation over $\set{1,\ldots,k}$ 
such that for all $1 \leq i \leq k-1$, we have $v_{f(i)} \in \Attr(v_{f(i+1)})$, 
we construct a winning strategy from $v$ that reaches $v_{f(1)}$, then $v_{f(2)}$, and so on. 
Note that this strategy only needs $k$ memory states.
Conversely, if $v \notin \cap_i \Attr(v_i)$, then Eve cannot win, as Adam can prevent her from reaching some reachability set.
	\item If the preorder $\preceq$ is not total, then there exist $v_i$ and $v_j$ 
such that $v_i \notin \Attr(v_j)$ and $v_j \notin \Attr(v_i)$.
In this case Adam wins from everywhere, following the strategy ``if $v_i$ or $v_j$ has been reached, then avoid the other''. 
Note that this strategy only needs $2$ memory states.
\end{itemize}
Checking that the preorder $\preceq$ is total can be done in polynomial time.
\hfill\qed
\end{proof}

Note that as a corollary, we get memory upper bounds in this case: Eve needs at most $k$ memory states and Adam at most $2$.
It is not difficult to see that these bounds are tight.

\subsection{Reachability sets of size two}

Let us now turn to the case where reachability sets have size two.
We first extend the technique used for the previous case: it was stated that
``Eve wins if and only if there is a total order on colored vertices''.
A similar approach works for one-player arenas, through a reduction to the satisfiability problem
of boolean formulas where clauses have size two.
(This latter problem is known to be decidable in polynomial time.)

\begin{theorem}[Generalized reachability one-player games where color appears twice]
Solving generalized reachability one-player games where reachability sets have size two is in $\PTIME$.
\end{theorem}

\begin{proof}
As in the previous subsection, we consider the preorder defined by $v \preceq v'$ if $v \in \Attr(v')$.
Note that in the case of one-player arenas, $v \in \Attr(v')$ reduces to ``there is a path from $v$ to $v'$''.

Let $F_i = \set{x_i,y_i}$ be the reachability sets, and $v_0$ be a starting vertex. 
We assume without loss of generality that there is a path from $v_0$ to every $F_i$ (that is, either to $x_i$ or $y_i$),
otherwise Eve cannot win. 
(This property is easily checked in deterministic polynomial time.)
A first statement is as follows: Eve wins from $v_0$ if and only if 
there exist $v_1,\ldots,v_k$ colored vertices such that
\begin{enumerate}
	\item for all $0 \leq i \leq k-1$, $v_i \preceq v_{i+1}$ and
	\item each color appears in $\set{v_1,\ldots,v_k}$.
\end{enumerate}
We turn this condition into a boolean formula where clauses have size $2$.
We consider the $2\cdot k$ variables $X_i$ and $Y_i$,
that correspond to vertices $x_i$ and $y_i$.
We define the formula $\phi$:
$$\underbrace{\bigwedge_{\ } \set{(\neg X \vee \neg Y) \mid 
\textrm{ if } x \not\preceq y \textrm{ and } y \not\preceq x}}_{(a)}
\wedge \underbrace{\bigwedge_{i} (X_i \vee Y_i)}_{(b)},$$
where $x,y$ ranges over colored vertices (that is, vertices from $F_i$ for some $i$).

We argue that Eve wins from $v_0$ if and only if $\phi$ is satisfiable.
Assume Eve wins from $v_0$: let $v_1,\ldots,v_k$ as in the previous statement,
and set the corresponding variables to true and the others to false, we claim that the formula $\phi$ is satisfied.
Indeed, condition 2. ensures that the clauses under-braced $(b)$ are satisfied,
and for the clauses under-braced $(a)$, let $x,y$ such that $x \not\preceq y$ and $y \not\preceq x$,
if $x$ is one of the $v_i$'s, then $y$ cannot be, so $\neg X \vee \neg Y$ holds.
Conversely, assume that $\phi$ is satisfiable.
The clauses under-braced $(a)$ ensures that the order $\preceq$ is total over vertices set to true.
The clauses under-braced $(b)$ ensures that at least one vertex from each reachability set is set to true.
Combining those two statements, we reach the condition stated above.

The latter allows to decide in polynomial time whether Eve wins from $v_0$ 
by checking the formula $\phi$ for satisfiability.\hfill\qed
\end{proof}

We do not know the exact complexity of generalized reachability games where reachability sets have size $2$.
In the remaining of this subsection, we discuss this question, focusing on memory requirements for both players.

The memory required for Eve is still exponential, as shown in Figure~\ref{fig:counter_ex_memory} for $k = 4$.
Specifically, it shows a generalized reachability game where reachability sets have size $2$ won by Eve, 
where she needs $2^{\ksurdeux + 1} - 1$ bits of memory to win.
The arena is divided into two parts: the left hand side is a flower with $\ksurdeux$ petals,
and the right hand side a one-player arena.
The game starts at the heart of the flower.
First, Eve asks for each petal a color.
Once this task is completed, she can move to the right hand side to reach the remaining colors.
Eve needs to remember the $\ksurdeux$ choices made by Adam (one for each petal), in order to reverse them:
if Adam chose the color $1$, then the color $2$ has not been reached,
so Eve has to choose color $2$.
Remembering those choices and asking for each petal requires $2^{\ksurdeux + 1} - 1$ memory states.
(This is the size of the complete binary tree of depth $\ksurdeux$.)

\begin{figure}
\begin{center}
\unitlength = 0.9mm
\begin{picture}(60,40)(0,0)
 	\gasset{Nw=6,Nh=6}

 	\node[Nmarks=i](0)(10,20){}
  
 	\rpnode[polyangle=45](p1)(10,40)(4,4){}
 	\node(1)(0,30){$1$}
 	\node(2)(20,30){$2$}

 	\drawedge(0,p1){}
 	\drawedge[curvedepth=3](p1,2){}
 	\drawedge[curvedepth=-3](p1,1){}
 	\drawedge[curvedepth=-3](1,0){}
 	\drawedge[curvedepth=3](2,0){}

 	\rpnode[polyangle=45](p2)(10,0)(4,4){}
 	\node(3)(0,10){$3$}
 	\node(4)(20,10){$4$}

 	\drawedge(0,p2){}
 	\drawedge[curvedepth=3](p2,3){}
 	\drawedge[curvedepth=-3](p2,4){}
 	\drawedge[curvedepth=3](3,0){}
 	\drawedge[curvedepth=-3](4,0){}

 	\node(c1)(30,20){}
 	\node(d1)(40,10){$1$}
 	\node(d2)(40,30){$2$}
 	\node(c2)(50,20){}
 	\node(d3)(60,10){$3$}
 	\node(d4)(60,30){$4$}

 	\drawedge(0,c1){}
 	\drawedge(c1,d1){}
 	\drawedge(c1,d2){}
 	\drawedge(d1,c2){}
 	\drawedge(d2,c2){}
 	\drawedge(c2,d3){}
 	\drawedge(c2,d4){}
	\drawloop[loopangle=0](d3){}
	\drawloop[loopangle=0](d4){}	
\end{picture}
\end{center}
\caption{A generalized reachability game where Eve needs $2^{\ksurdeux + 1} - 1$ memory states to win}
\label{fig:counter_ex_memory}
\end{figure}

On the other hand, the exact memory required for Adam remains open.
The figure~\ref{fig:adamsize2}, following an idea of Christof Loeding, 
shows a generalized reachability game where reachability sets have size $2$ won by Adam, 
where he needs $4$ memory states to win.
The game starts from the left hand side vertex. First Eve chooses and visits three of the four colors 
(two colors in the first column, $1$ and $2$ or $3$ and $4$, and then one in the second column),
and sends the pebble to the right hand side vertex, controlled by Adam.
There, he has four options, each allowing all colors but one.
Remembering the four possibilities requires four memory states, and leads to a win.
However, with less memory states, one of the four option will never be played, and Eve wins.

Quite surprisingly, we could not generalize this example to obtain a better lower bound than $4$.
We do not know whether this bound is tight (in any arena, if Adam wins, then he has a winning strategy with $4$ memory states),
which is plausible. 
Note that this would imply a $\cocompl\NP^\NP$ algorithm: guess a winning strategy for Adam with $4$ memory states, 
and compose this strategy with the game, then solve the resulting one-player game.

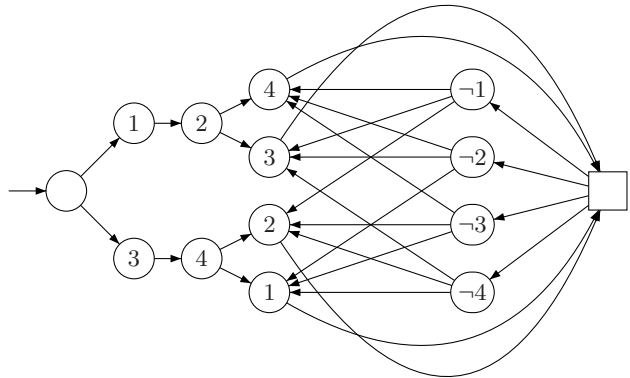
\begin{figure}
\begin{center}
\unitlength = 0.9mm
\begin{picture}(60,50)(0,-5)
 	\gasset{Nw=6,Nh=6}

 	\node[Nmarks=i](0)(0,20){}

 	\node(1)(10,30){$1$}
 	\node(2)(20,30){$2$}
 	\node(3)(10,10){$3$}
 	\node(4)(20,10){$4$}

 	\node(4b)(30,35){$4$}
 	\node(3b)(30,25){$3$}
 	\node(2b)(30,15){$2$}
 	\node(1b)(30,5){$1$}
  
	\gasset{Nadjust=w}
 	\node(234)(60,35){$\neg 1$}
 	\node(134)(60,25){$\neg 2$}
 	\node(124)(60,15){$\neg 3$}
 	\node(123)(60,5){$\neg 4$}

 	\rpnode[polyangle=45](c)(80,20)(4,4){}

 	\drawedge(0,1){}
 	\drawedge(0,3){}

 	\drawedge(1,2){}
 	\drawedge(3,4){}

 	\drawedge(2,4b){}
 	\drawedge(2,3b){}
 	\drawedge(4,2b){}
 	\drawedge(4,1b){}

 	\drawedge[curvedepth=15](4b,c){}
 	\drawedge[curvedepth=25](3b,c){}
 	\drawedge[curvedepth=-25](2b,c){}
 	\drawedge[curvedepth=-15](1b,c){}

 	\drawedge(c,123){}
 	\drawedge(c,124){}
 	\drawedge(c,134){}
 	\drawedge(c,234){}

 	\drawedge(123,1b){}
 	\drawedge(123,2b){}
 	\drawedge(123,3b){}

 	\drawedge(124,1b){}
 	\drawedge(124,2b){}
 	\drawedge(124,4b){}

 	\drawedge(134,1b){}
 	\drawedge(134,3b){}
 	\drawedge(134,4b){}

 	\drawedge(234,2b){}
 	\drawedge(234,3b){}
 	\drawedge(234,4b){}
\end{picture}
\end{center}
\caption{A generalized reachability game where Adam needs $4$ memory states to win}
\label{fig:adamsize2}
\end{figure}

\medskip\noindent{\bf Open problems.}
We were not able to give the exact complexity of generalized reachability games where reachability sets have size $2$.
The memory approach, showing that Adam has winning strategy of constant size, seems promising towards this question.

\bibliography{people,short,papers}
\bibliographystyle{alpha}

\end{document}